\documentclass{article}
\RequirePackage{amsthm,amsmath,amsbsy,amssymb,bm,color,graphicx,hyperref}

\newcommand{\I}{\ensuremath{\operatorname{\mathbb{I}}}}
\renewcommand{\P}{\ensuremath{\operatorname{\mathbb{P}}}}
\newcommand{\E}{\ensuremath{\operatorname{\mathbb{E}}}}
\newcommand{\var}{\ensuremath{\operatorname{Var}}}
\newcommand{\cov}{\ensuremath{\operatorname{Cov}}}

\let\originalleft\left
\let\originalright\right
\renewcommand{\left}{\mathopen{}\mathclose\bgroup\originalleft}
\renewcommand{\right}{\aftergroup\egroup\originalright}

\theoremstyle{plain}
\newtheorem{Theorem}{Theorem}
\newtheorem{lemma}{Lemma}
\newtheorem{proposition}{Proposition}

\begin{document}

\title{Network histograms and universality of blockmodel approximation}

\author{Sofia C.~Olhede and Patrick J.~Wolfe}

\maketitle

\begin{abstract}
In this article we introduce the network histogram: a statistical summary of network interactions, to be used as a tool for exploratory data analysis. A network histogram is obtained by fitting a stochastic blockmodel to a single observation of a network dataset. Blocks of edges play the role of histogram bins, and community sizes that of histogram bandwidths or bin sizes.  Just as standard histograms allow for varying bandwidths, different blockmodel estimates can all be considered valid representations of an underlying probability model, subject to bandwidth constraints.  Here we provide methods for automatic bandwidth selection, by which the network histogram approximates the generating mechanism that gives rise to exchangeable random graphs.  This makes the blockmodel a universal network representation for unlabeled graphs. With this insight, we discuss the interpretation of network communities in light of the fact that many different community assignments can all give an equally valid representation of such a network.  To demonstrate the fidelity-versus-interpretability tradeoff inherent in considering different numbers and sizes of communities, we analyze two publicly available networks---political weblogs and student friendships---and discuss how to interpret the network histogram when additional information related to node and edge labeling is present.

\vspace{\baselineskip}%
\end{abstract}

The purpose of this article is to introduce the network histogram---a nonparametric statistical summary obtained by fitting a stochastic blockmodel to a single observation of a network dataset.  A key point of our construction is that it is not necessary to assume the data to have been generated by a blockmodel.  This is crucial, since networks provide a general means of describing relationships between objects.  Given $n$ objects under study, a total of $\binom{n}{2}$ pairwise relationships are possible.  When only a small fraction of these relationships are present---as is often the case in modern high-dimensional data analysis across scientific fields---a network representation simplifies our understanding of this dependency structure.

One fundamental characterization of a network comes through the identification of community structure~\cite{girvan2002community}, corresponding to groups of nodes that exhibit similar connectivity patterns. The canonical statistical model in this setting is the stochastic blockmodel~\cite{holland1983stochastic}: it posits that the probability of an edge between any two network nodes depends only on the community groupings to which those nodes belong.   Grouping nodes together in this way serves as a natural form of dimensionality reduction: as $n$ grows large, we cannot retain an arbitrarily complex view of all possible pairwise relationships.  Describing how the full set of $n$ objects interrelate is then reduced to understanding the interactions of $k \ll n$ communities.  Studying the properties of fitted blockmodels is thus  important~\cite{bickel2009nonparametric, zhao2011community}.

Despite the popularity of the blockmodel, and its clear utility, scientists have observed that it often fails to describe all the structure present in a network~\cite{airoldi2008mixed, karrer2011stochastic, newman2011communities, gopalan2013efficient}.  Indeed, as a network becomes larger, it is no longer reasonable to assume that a majority of its structure can be explained by a blockmodel with a fixed number of blocks.  Extensions to the blockmodel have focused on capturing additional variability, for example through mixed community membership~\cite{airoldi2008mixed} and degree correction~\cite{karrer2011stochastic, zhao2012consistency}.  However, the simplest and most natural method of extending the descriptiveness of the blockmodel is to add blocks, so that $k$ grows with $n$.  As more and more blocks are fitted, we expect an increasing degree of structure in the data to be explained. The natural questions to ask then are many:  What happens as we fit more blocks to an arbitrary network dataset, if the true data-generating mechanism is not a blockmodel? At what rate should we increase the number of blocks used, depending on the variability of the network?  We discuss these and other questions in this article.

We will stipulate how the dimension $k$ of the fitted blockmodel should be allowed to increase with the size $n$ of the network.  This increase will be dictated by a tradeoff between the sparsity of the network and its heterogeneity or smoothness.  If one assumes that a $k$-community blockmodel is the actual data-generating mechanism, then theory has already been developed which allows $k$ to grow with $n$~\cite{rohe2011spectral, choi2012stochastic, chatterjee2012matrix}, and methods have been suggested for choosing the number of blocks based on the data~\cite{fishkind2013consistent, bickel2013hypothesis}.  General theory for the case when the blockmodel is merely \emph{approximating} the observed network structure is nascent, with~\cite{choi2012co} treating the case of dense bipartite graphs with a fixed number of blocks, and~\cite{wolfe2013} establishing the first such results for the setting of relevance here.

\section{From stochastic networks to histograms}\label{exchangeable}

\subsection{A simple stochastic network model}

We encode the relationships between $n$ objects using $\binom{n}{2}$ binary random variables. Each of these variables indicates the presence or absence of an edge between two nodes, and can be collected into an $n \times n$ adjacency matrix $A$, such that $A_{ij} = 1$ if nodes $i$ and $j$ are connected, and $A_{ij} = 0$ otherwise, with $A_{ii} = 0$.  This yields what is known as a simple random graph.

Models for unlabeled graphs are strongly related to the statistical notion of exchangeability, a fundamental concept describing random variables whose ordering is without information. To relate to exchangeable variables, we appeal to the Aldous--Hoover theorem~\cite{bickel2009nonparametric}, and model our network hierarchically using three components:
\begin{enumerate}
\item A fixed, symmetric function $f(x,y)$ termed a graphon~\cite{lovasz2012large}, which behaves like a probability density function for $0 < x,y < 1$;
\item For each $n$, a random sample $\xi$ of $n$ uniform random variables $\{ \xi_1, \ldots, \xi_n\}$ which will serve to index the graphon $f(x,y)$; and
\item For each $n$, a deterministic scaling constant $\rho_n > 0$, specifying the expected fraction of edges $ \smash{ \binom{n}{2}^{-1} \E \sum_{i<j} A_{ij} } $ in the network.
\end{enumerate}
For each $n$, our simple stochastic network model is then
\begin{equation}
\label{eqn:Aij}
A_{ij} \,\vert\, \xi_i, \xi_j \sim \operatorname{Bernoulli}\bigl( \, \rho_n f\left(\xi_i,\xi_j\right) \, \bigr), \quad 1 \leq i < j \leq n,
\end{equation}
where for statistical identifiability of $\rho_n$ we assume
\begin{equation}
\label{eq:f-norm}
\textstyle \iint_{(0,1)^2} f\left(x,y\right) \,dx \,dy = 1.
\end{equation}

In this way we model the network structure itself---rather than the particular ordering in which the network's nodes are arranged in $A$.  As an example, Fig.\ \ref{fig:polblogsdata} shows three different orderings of the adjacency matrix of a network of US political weblogs recorded in 2005~\cite{adamic2005}, each emphasizing a different aspect of the network.
\begin{figure*}[t]
\hspace{-3.33cm}\includegraphics[width=1.5\columnwidth]{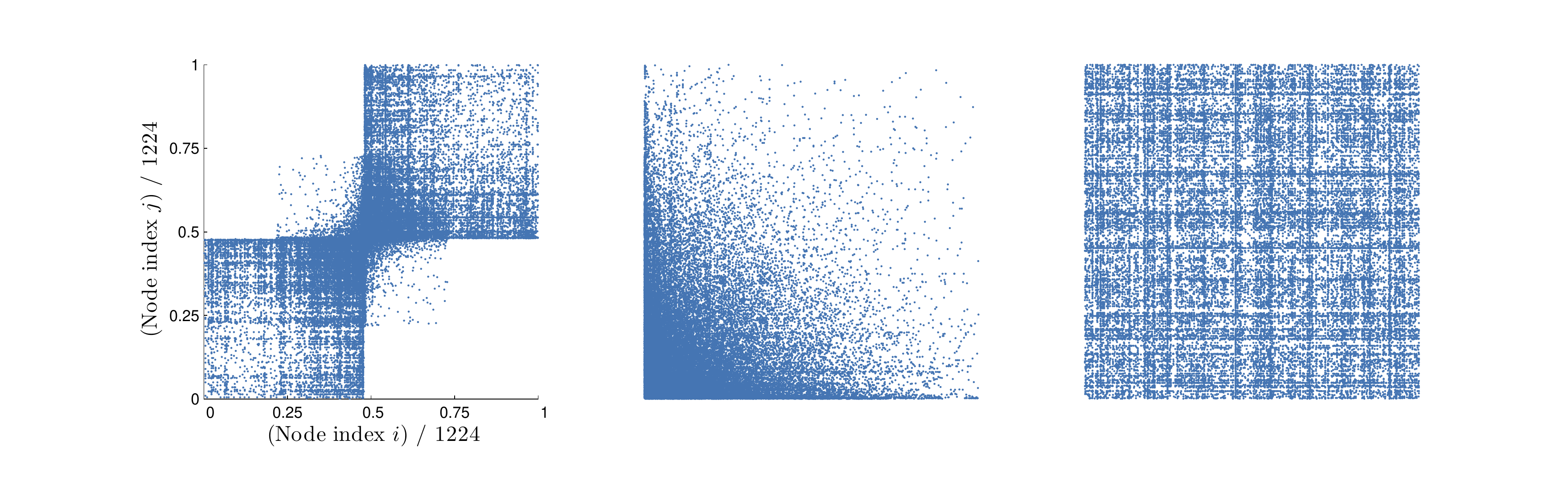}
\caption{\label{fig:polblogsdata} Three adjacency matrix representations of the political weblog data of~\cite{adamic2005}, each showing all 1224 blogs with at least one link to another blog in the dataset (links denoted by blue dots).  The first 586 blogs in the leftmost panel are categorized by~\cite{adamic2005} as liberal and the remaining 638 as conservative; note the sparsity of cross-linkages. The middle panel shows the same data, ordered by decreasing number of links, while the third panel shows how a random labeling obscures structure.}
\end{figure*}

We see from this generative mechanism that any (symmetric) re-arrangement of the $x$- and $y$-axes of $f$ will lead to the same probability distribution on unlabeled graphs, and in fact a graphon describes an entire equivalence class of functions.  We assume that at least one member of this equivalence class is H\"older continuous, which we refer to as $f$ without loss of generality; that $f$ is bounded away from 0 and $\rho_n f$ is bounded away from 1; and that the sequence $\rho_n$ is monotone non-increasing and decays more slowly than $ n^{-1} \log^3 n $, so that the average network degree grows faster than $\log^3 n$.

To summarize the network we therefore wish to estimate the graphon $f\left(x,y\right)$, up to re-arrangement of its axes. By inspection,
\begin{equation*}
\P\left(A_{ij}=1\right) = \E A_{ij} = \rho_n \textstyle \iint_{(0,1)^2} f\left(x,y\right) \,dx \,dy = \rho_n,
\end{equation*}
and so we may estimate $\rho_n$ via the sample proportion estimator
\begin{equation}
\label{eq:rhohat}
\hat \rho_n = \tbinom{n}{2}^{-1} \textstyle \sum_{i<j} A_{ij} .
\end{equation}

\subsection{The network histogram}

Given a single adjacency matrix $A$ of size $n \times n$, we will estimate $f\left(x,y\right)$ (up to re-arrangement of its axes) using a stochastic blockmodel with a single, pre-specified community size $h$, to yield a network histogram.  Choosing the bandwidth $h$ is equivalent to choosing a specific number of communities $k$---corresponding to the number of bins in an ordinary histogram setting.

To define the network histogram, we first write the total number of network nodes $n$ in terms of the integers $h$, $k$, and $r$ as $n  = h k + r$, where $k = \left\lfloor n / h \right\rfloor$ is the total number of communities; $h$ is the corresponding bandwidth, ranging from 2 to $n$; and $r = n \mod h $ is a remainder term between 0 and $h-1$.  To collect together the nodes of our network that should lie in the same group, we introduce a community membership vector $z$ of length $n$.  All components of $z$ will take values in $\{1,\dots,k\}$, and will share the same values whenever nodes are assigned to the same community.

The main challenge in forming a network histogram lies in estimating the community assignment vector $z$ from $A$.  To this end, for each $n$, let the set $\mathcal{Z}_k \subseteq \{1,\ldots,k\}^n $ contain all community assignment vectors $z$ that respect the given form of $n  = h k + r$.  Thus $\mathcal{Z}_k$ consists of all vectors $z$ with $h$ components equal to each of the integers from 1 to $k-1$ (up to relabeling), and $h+r$ components equal to $k$ (again, up to relabeling).  In this way, $\mathcal{Z}_k$ indexes all possible histogram arrangements of network nodes into $k-1$ communities of equal size $h$, plus an additional community of size $h+r$.

Many ways of estimating $z$ from a single observed adjacency matrix $A$ have been explored in the literature.  In essence, nodes that exhibit similar connectivity patterns are likely to be grouped together (an idea that can be exploited directly if multiple observations of the same network are available; see~\cite{airoldi2013stochastic}).  We can formalize this notion through the method of maximum likelihood, by estimating
\begin{equation}
\label{eq:MPLE}
\! \hat z= \operatornamewithlimits{argmax}_{ z \in \mathcal{Z}_k } \sum_{ i < j } \left\{ A_{ij} \log \bar A_{z_i z_j} + \left( 1 \!-\! A_{ij} \right) \log \left( 1 \!-\! \bar A_{z_i z_j} \right) \right\} \!  , \!\!\!
\end{equation}
where for all $ 1 \leq a,b \leq k$ we define the histogram bin heights
\begin{equation}
\label{eq:Abar}
\bar A_{ab} = \frac{ \sum_{ i < j } A_{ ij } \I\left( \hat z_i = a \right) \I\left( \hat z_j = b \right) }{ \sum_{i<j} \I\left( \hat z_i = a \right) \I\left( \hat z_j = b \right) }.
\end{equation}
Each bin height $\bar A_{ab}$ is the proportion of successes (edges present) in the histogram bin corresponding to a block of Bernoulli trials, with the grouping of nodes into communities determined by the objective function in~\eqref{eq:MPLE}.  Since $A$ is symmetric, we have $\bar A_{ab} = \bar A_{ba}$.

Combining~\eqref{eq:rhohat} and~\eqref{eq:Abar}, we obtain our network histogram:
\begin{equation}
\label{eq:hist-est}
\! \hat f\left( x , y ; h \right) = \hat \rho_n^+ \bar A_{ \min( \lceil n x/ h\rceil, k) \min( \lceil n y/ h\rceil, k) }  , \,\,\, 0 < x,y < 1, \!\!
\end{equation}
with $ \hat \rho_n^+ $ the generalized inverse of $ \hat \rho_n $.

\section{Universality of blockmodel approximation}

\subsection{Blockmodel approximations of unlabeled graphs}

To understand the performance of blockmodel approximation, we must compare $\smash{ \hat f }$ to $f$ in a way that is invariant to all symmetric re-arrangements of the axes of $f$.  We will base our comparison on the graph-theoretic notion of cut distance, which in mathematical terminology defines a compact metric space on graphons~\cite{lovasz2012large}.  Just as our notion of unlabeled graphs treats any two adjacency matrices as the same if one can be obtained by symmetrically permuting the rows and columns of the other, we will compare two graphons via an invertible, symmetric rearrangement of the $x$ and $y$ axes that relates one graphon to the other.  We call $\mathcal{M}$ the set of all such rearrangements---formally, it is the set of all measure-preserving bijections of the form $[0,1] \rightarrow [0,1]$.

In~\cite{wolfe2013} we formulated convergence rates at which the resulting error between $\smash{ \hat f }$ and $f$ shrinks to zero as $ n \rightarrow \infty $ under the assumptions above.  Here we consider mean integrated square error (MISE), typically used in standard histogram theory (see, e.g.,~\cite{tsybakov2009introduction}), and take its greatest lower bound over all possible rearrangements $\sigma \in \mathcal{M}$:
\begin{equation}
\label{eq:MISE}
\!\operatorname{MISE}\bigl( \hat f \bigr) = \E \inf_{\sigma\in \mathcal{M} } \! \iint_{(0,1)^2} \!\!\!\!\!\!\!\! \bigl| f\left( x,y \right) - \hat f\bigl( \sigma(x),\sigma(y) ; h \bigr) \bigr|^2 \, dx \, dy .\!\!\!
\end{equation}
This definition factors out the unknown ordering of the data $A$ induced by $ \{ \xi_1, \ldots, \xi_n\} $ in the model of~\eqref{eqn:Aij}, accounting for the fact that $A$ may represent an unlabeled graph. The appearance of $\sigma$ may at first seem counterintuitive, but its introduction is necessary once we use~\eqref{eqn:Aij} to model $A$. In contrast to the optimization of~\eqref{eq:MISE} over all $\sigma \in \mathcal{M}$, which is purely conceptual, the vector $\hat z$ results from the algorithmic optimization of~\eqref{eq:MPLE} given an observed adjacency matrix $A$, and determines which entries of $A$ are averaged to estimate $f(x,y)$.

Using a single bandwidth $h$ to form $ \smash{ \hat f\left( x , y ; h \right) } $ in~\eqref{eq:MISE} represents a conceptual paradigm shift away from the standard usage of the stochastic blockmodel.  Instead of representing community structure, a blockmodel can be used as a universal mechanism to represent an arbitrary unlabeled network. In practice, of course, we may well have information that implies certain labelings or orderings of the network nodes.  The assumption of exchangeability models our ignorance of this information as a baseline, just as we may choose to cluster a Euclidean dataset without taking into account any accompanying labels.  Thus we require our error metric to respect this ignorance, even if we later choose to interpret a fitted histogram in light of node labels (as one might with Euclidean data clusters, and as we shall do below).

The goal in the setting of exchangeable networks is therefore no longer to discover latent community structure, but rather simply to group together nodes whose patterns of interactions are similar.  Thus the interpretation of the fitted groups has altered. Instead of uncovering true underlying communities that might have given rise to the data, our blocks now approximate the generative process, up to a resolution chosen by the user---namely the bandwidth, $h$.  This can be related to previous understanding of the error behavior when the data are generated by a blockmodel, both in the regimes of $\rho_n$ corresponding to growing degrees~\cite{choi2012stochastic} as well as even sparser ones~\cite{krzakala2013spectral}.

\subsection{The oracle network labeling}

We next show how the ideal or \emph{oracle} labeling information, were it to be available, would yield the optimal bandwidth parameter $h$ for any given network histogram.  This oracle information arises from the latent random variables $\{ \xi_1, \ldots, \xi_n\}$ present in the generative model of~\eqref{eqn:Aij}.  In this setting, instead of fitting blocks of varying sizes to the network, to be interpreted as community structure, we rely on the fact that the simplest type of blockmodel will suffice, with only a single tuning parameter $h$.  The existence of a smooth limiting object---namely the graphon $f\left(x,y\right)$---implies that a \emph{single} community size or bandwidth will provide an adequate summary of the entire network.

To choose $h$, we therefore employ the notion of a network oracle.  As in standard statistical settings~\cite{tsybakov2009introduction}, the oracle provides information that is not ordinarily available, thereby serving to bound the performance of any data-driven estimation procedure.  The oracle estimator for each histogram bin height takes the same form as~\eqref{eq:Abar}, but uses a unique (almost surely) labeling $\tilde z$ calculated from the latent random vector $\xi$.  This labeling is given by $ \smash{ \tilde z_i = \min\left\{ \lceil (i)^{-1} / h \rceil, k \right\} } $, where $\left(i\right)^{-1}$ is the rank, from smallest to largest, of the $i$th element of $\xi$.  Thus, $ \tilde z $ orders elements of the unobserved vector $\xi$, sorts the indices of the data according to this ordering, and then groups these indices into sets of size $h$, with one additional set of size $h+r$.

With the oracle labeling $\tilde z$, we may define the graphon oracle estimator from the block averages $\bar A^*_{ab}$ according to
\begin{align}
\nonumber
\bar A^*_{ab} &= \frac{ \sum_{ i<j} A_{ ij } \I\left(\tilde z_i=a\right) \I\left(\tilde z_j=b\right) }{  \sum_{ i<j} \I\left(\tilde z_i=a\right) \I\left(\tilde z_j=b\right) } ,
\\
\label{graph-est}
\hat f^*\left(x,y; h\right) &= \rho_n^{-1} \bar A^*_{ \min( \lceil n x/ h\rceil, k) \min( \lceil n y/ h\rceil, k) } .
\end{align}
Comparing~\eqref{graph-est} with its counterpart in~\eqref{eq:hist-est}, we see that the oracle serves to replace the estimators of~\eqref{eq:rhohat} and~\eqref{eq:MPLE} with their ideal quantities. Thus the oracle estimator is based on a priori knowledge of the sparsity parameter $\rho_n$ and the latent vector $\xi$.  In this sense, it shows the best performance that can be achieved for a fixed bandwidth $h$, by providing knowledge of the scaling and ordering necessary for the estimator to become a linear function of the data.

\section{Determining the histogram bandwidth}

\subsection{Oracle mean-square error bound}

By making use of the network oracle, we can determine what performance limits are possible, and in turn derive a rule of thumb for selecting the bandwidth $h$.  We assume here that $f$ is differentiable, noting that this result extends to H\"older continuous functions, as shown in Appendix~\ref{app}.

\smallskip
\begin{Theorem}[Network histogram oracle bandwidth selection]\label{MISE}
Assume that $h$ grows more slowly than $n$, and that the graphon $f\left(x,y\right)$ is differentiable, with a gradient magnitude bounded by $M$.  Then as $n$ grows the oracle mean integrated square error satisfies the bound
\begin{equation*}
\!\!\!\! \operatorname{MISE}\bigl( \hat f^* \bigr)
\le M^2 \left\{ 2 \left( \tfrac{h}{n} \right)^{ 2 } +  \tfrac{1}{n}  + \tfrac{1}{M^2} \left( \tfrac{1}{h^2\rho_n} \right) \right\} \left\{ 1 + o(1) \right\} . \!\!\!\!
\end{equation*}
The right-hand side of this expression is minimized by setting $h \!=\! h^*$:
\begin{equation}
\label{opt-block}
 h^* = (2 M^2 \rho_n)^{-1/4} \cdot \sqrt{n},
\end{equation}
whence $ \operatorname{MISE}\bigl( \hat f^* \bigr) $ evaluated at $h^*$ decays at the rate $ \smash{ 1 / \sqrt{ \binom{n}{2} \rho_n } } $:
\begin{equation}
\label{eq:orac-mise-hstar}
\operatorname{MISE}\bigl( \hat f^* \bigr) \Bigr|_{h = h^*}
\!\!\!\! \!\!\!\! \!\!\! \le M^2 \left[ \tfrac{2}{M} \left\{ \tbinom{n}{2} \rho_n \right\}^{-1/2} \!\! +  \tfrac{1}{n}  \right] \left\{ 1 + o(1) \right\} \! . \!
\end{equation}
\end{Theorem}

\begin{proof}
We evaluate~\eqref{eq:MISE} with $\hat f$ set equal to $\hat f^*$ as defined in~\eqref{graph-est}, and with $\sigma(x)$ set equal to $x$ to obtain an upper bound on the error criterion $\smash{\operatorname{MISE}\bigl( \hat f^* \bigr)}$.  This yields the bias--variance decomposition
\begin{multline*}
\operatorname{MISE}\bigl( \hat f^* \bigr)
\le \E \iint_{(0,1)^2} \bigl| f\left( x,y \right) - \hat f^*\left( x,y ; h \right) \bigr|^2 \, dx \, dy
\\ = \sum_{a,b=1}^k \iint_{\omega_{ab}} \left\{ \bigl| f\left(x,y\right) - \rho_n^{-1} \E\bar A^*_{ab} \bigr|^2 + \rho_n^{-2} \var \bar A^*_{ab} \right\} \, dx \, dy ,
\end{multline*}
with $\omega_{ab}$ the domain of integration corresponding to the block $ \bar A_{ab}$.

Now let $ \bar f_{ab}= \left| \omega_{ab} \right|^{-1} \smash{ \iint_{\omega_{ab}} } f\left(x,y\right) \,dx \,dy $ be the average value of $f$ over $\omega_{ab}$, and $ \smash{ \overline{ f_{ab}^2 } } $ the average value of $f^2$.  Using the assumed smoothness of $f$ in a manner quantified by Proposition~1 in Appendix~\ref{app}, we substitute for $\var \bar A^*_{ab}$ and $\E \bar A^*_{ab}$ to obtain
\begin{align*}
\!\!\!\! &\operatorname{MISE}\bigl( \hat f^* \bigr)
\le \! \sum_{a,b=1}^k \! \iint_{\omega_{ab}} \!\! \left[ \vphantom{\frac{\bar{f}_{ab}-\rho_n\overline{f^2}_{ab}}{\rho_n h_{ab}^2}} \left| \left\{ f\left(x,y\right) \!-\! \bar f_{ab} \right\} \!+\! \left\{ \bar f_{ab} \!-\! \rho_n^{-1} \E\bar A^*_{ab} \right\} \right|^2 \right.
\\
& \qquad \qquad \qquad \,\, + \left. \frac{\bar{f}_{ab}-\rho_n\overline{f^2}_{ab}}{\rho_n h_{ab}^2} +
 \frac{M\left\{1+o(1)\right\}}{\rho_n h_{ab}^2 (2n)^{1/2}}+\frac{M^2}{2n} \right] \, dx \, dy\\
& \le \sum_{a,b=1}^k \left[  \iint_{\omega_{ab}} \!\!\!\!\!\left| f\left(x,y\right) - \bar f_{ab} \right|^2\,dx\,dy\!+ \left\{\frac{M^2\left\{1+o(1)\right\}}{2n}\right.\right.\\
& \qquad \qquad \qquad \,\, \left. + \frac{\bar{f}_{ab}-\rho_n\overline{f^2}_{ab}}{\rho_n h_{ab}^2} +
 \frac{M\left\{1+o(1)\right\}}{ (2n)^{1/2}} \frac{1}{\rho_n  h_{ab}^2}+\frac{M^2}{2n} \right\} \\
& \qquad \qquad \qquad \,\,\,\, \cdot \left. \frac{ \left\{h+r \I\left(a=k\right)\right\} \left\{h+r \I\left(b=k\right)\right\}}{n^2} \right] ,
\end{align*}
with $ \smash{ h_{ab}^2 = \sum_{i<j} \I\left(\tilde z_i=a\right) \I\left(\tilde z_j=b\right) }$.  Applying Lemma~1 in Appendix~\ref{app} to each $\iint_{\omega_{ab}} \left| f\left(x,y\right) - \bar f_{ab} \right|^2\,dx\,dy$,
\begin{align*}
\sum_{a,b=1}^k \iint_{\omega_{ab}} \left| f\left(x,y\right) - \bar f_{ab} \right|^2 \,dx \,dy \leq M^2 \cdot 2 \left( \tfrac{h}{n} \right)^{ 2 } \left\{ 1 + \mathcal{O}\left( \tfrac{h}{n} \right) \right\},
\end{align*}
with the $\mathcal{O}(h/n)$ term due to the grouping of size $h+r$. Using~\eqref{eq:f-norm},
\begin{multline*}
\sum_{a,b=1}^k \frac{\bar{f}_{ab}}{\rho_n h_{ab}^2 } \frac{ \left\{h+r \I\left(a=k\right)\right\} \left\{h+r \I\left(b=k\right)\right\}}{n^2}
\\ = \sum_{a,b=1}^k \frac{1}{\rho_n h_{ab}^2}\iint_{\omega_{ab}}  f\left(x,y\right) \,dx \,dy
= \frac{1}{\rho_n h^2} \left\{ 1 + o(1) \right\}.
\end{multline*}
Combining these simplifications yields the stated expression.
\end{proof}

This theorem informs the selection of a network histogram bandwidth $h$.  It quantifies how the oracle integrated mean square error depends on the smoothness of the graphon $f$, relative to the size and sparsity of the observed adjacency matrix $A$.  The theorem decomposes this error into three contributions: smoothing bias, which scales as $ M^2 \,(h/n)^2 $; resolution bias, which scales as $ M^2 / n $; and variance contributions, which scale as the inverse of the effective degrees of freedom $h^2 \rho_n$ of each bin.  As shown in~\cite{wolfe2013}, ensuring that $h^2 \rho_n$ grows faster than $ \log^3 n $ will enable consistent estimation of the graphon when $z$ is estimated according to~\eqref{eq:MPLE}; this accounts for the additional variance involved in estimating $z$ in the non-oracle setting.

Theorem~\ref{MISE} subsequently enables us to choose a bandwidth $h$ that respects the global properties of the network.  If we were to know $\rho_n$ and $M$, then the theorem provides directly for an oracle choice of bandwidth $h^*$ according to~\eqref{opt-block}.  From this expression we see that for the case of a dense network, with $ \rho_n \propto 1$, the oracle choice of bandwidth $h^*$ scales as $\sqrt{n}$.  More generally, we observe that as the sparsity of the network increases, $h^*$ must also increase, while as the gradient magnitude of the graphon increases, $h^*$ must decrease.  If $f$ is not differentiable but is still H\"older continuous, then the H\"older exponent will appear in the theorem expressions, leading to a smaller bandwidth for a given $n$ and $\rho_n$.

Finally, Theorem~\ref{MISE} provides for an upper bound on the oracle mean integrated square error when the network histogram bandwidth is set equal to $h^*$.  This bound reveals the best possible estimation performance we might achieve for given values of $n$, $\rho_n$, and $M$.

\subsection{Automatic bandwidth selection}

Theorem~\ref{MISE} is important for our theoretical understanding of the bandwidth selection problem, as it shows the tradeoffs between sparsity, smoothness, and sample size.  It suggests that $h$ should grow at a rate proportional to $ \smash{ \rho_n^{-1/4} \sqrt{n} } $, with $\rho_n$ estimated via~\eqref{eq:rhohat}, and with a constant of proportionality depending on the squared magnitude $M^2$ of the graphon gradient.

To estimate $M^2$ from $A$, we will form a simple one-dimensional approximation of the graphon $f$ using the vector $d$ of sorted degrees.  This yields a nonparametric estimator for what is referred to as the canonical version of $ \smash{ \int_0^1 f(x,y) \,dy }$~\cite{bickel2009nonparametric}.  Whenever the smoothness of this canonical marginal is equivalent to that of $f$, then this procedure yields a suitable estimator $\smash{ \widehat{M^2} }$ according to the steps below.  In some instances, however, the marginal may be smoother than $f$; for example, let $B(x)$ denote the distribution function of a $\operatorname{Beta}(a,b)$ random variable, and suppose $f(x,y) \propto B^{-1}(x) B^{-1}(y) + B^{-1}(1\!-\!x) B^{-1}(1\!-\!y)$.  Then the marginal is constant, but the corresponding $M^2$ (and indeed the H\"older regularity of $f$) will depend on $a$ and $b$.

To proceed, assume that the rows and column of $A$ have been re-ordered such that $ \smash{ d_i = \sum_{j\neq i} A_{ij} } $ is increasing with $i$.  Enumerating the sampled elements $f(\xi_i,\xi_j)$ of the graphon in a $n \times n$ matrix $F$ under this same re-ordering, we obtain in analogy to~\eqref{eq:hist-est} a rank-one estimate of the sampled graphon as $ \smash{ \hat F \propto \hat \rho_n^+ d d^T\! } $.  Minimizing the Frobenius norm $ \smash{ \|\hat F-\hat \rho_n^{+} A\| }$ then leads to the expression $ \smash{ \hat F = [ \{ (d^T\!d )^{+} \}^{2} \hat \rho_n^{+} d^T\! A d ] d d^T\! } $.

We then use $ \smash{ \hat F } $ to estimate the bandwidth $h$ as follows:
\vspace{-0.25\baselineskip}%
\begin{enumerate}
\item Compute the vector $d$ of degrees of $A$; sort its entries.
\item Estimate the slope of the ordered $d$ over indices $\lfloor n/2 \rfloor \pm \lfloor c\sqrt{n}\rfloor $ for some choice of $c$; normally $c=4$ is appropriate. Treating the ordered entries of $d$ near $ \lfloor n/2 \rfloor $ as a set of observations, fit a line with slope $m$ and intercept $b$ using the system of equations
\begin{equation*}
d_{ \lfloor n/2 \rfloor +j} = jm+b ,
\quad j =-\lfloor c\sqrt{n}\rfloor, -\lfloor c\sqrt{n}\rfloor + 1, \dots, \lfloor c\sqrt{n} \rfloor.
\end{equation*}
By the method of least squares, this yields estimates $\hat m$ and $\hat b$.
\item Define the vector-valued function of first differences
\begin{equation*}
\Delta f\left(x,y\right) =
\begin{pmatrix}
f\left(x,y\right)-f\left(x+\frac{1}{n+1},y \right) \\
f\left(x,y\right)-f\left(x,y+\frac{1}{n+1}\right)
\end{pmatrix},
\end{equation*}
leading to the following gradient estimate:
\begin{equation*}
\widehat{\Delta f} = [ \{ (d^T\!d )^{+} \}^{2} \hat \rho_n^{+} d^T\! A d ] \, \begin{pmatrix}
\hat m\hat b &
\hat m\hat b
\end{pmatrix}^T\!.
\end{equation*}
Via $\| \widehat{\Delta f} \|^2$, we estimate the average squared magnitude of $\Delta f$:
\begin{equation}
\widehat{M^2}
\label{eq:Mhatsqrd-est} =
2 n^2 \{ (d^T\!d )^{+} \}^{4} (\hat \rho_n^{+})^2 (d^T\! A d)^2 \hat m^2\hat b^2
\left\{ 1 + o(1) \right\} .
\end{equation}
\item Substituting $\widehat{M^2}$ into~\eqref{opt-block}, we obtain the bandwidth estimate
\begin{equation}
\label{eq:bw-rule}
\!\!\!\! \widehat{h^*}
=(2 \widehat{M^2} \hat\rho_n)^{-\frac{1}{4}} \sqrt{n}
= \bigl(2 \{ (d^T\!d )^{+} \}^{2} d^T\! A d \cdot \hat m\hat b \bigr)^{-\frac{1}{2}} \hat \rho_n^{\frac{1}{4}} .
\end{equation}
\end{enumerate}
Equipped with this rule of thumb for selecting the bandwidth $h$, we can now calculate the network histogram $ \smash{ \hat f\bigl(x,y; \widehat{h^*} \bigr) }$.\vspace{-0.75\baselineskip}%

\section{Data analysis using network histograms}
\label{sec:simexamples}

Data analysis software to calculate the network histogram is available at the site \url{https://github.com/p-wolfe/network-histogram-code}.

\subsection{Political weblog data}

To demonstrate the utility of the network histogram, we first analyze a well-studied dataset of political weblogs described in~\cite{adamic2005} and illustrated in Fig.\ \ref{fig:polblogsdata}.  This dataset was collected to quantify the degree of interaction between liberal and conservative blogs around the time of the 2004 US presidential election, and consists of a snapshot of nearly 1500 weblogs from February 8, 2005.  An edge is considered to be present between two blogs whenever at least one of the blogs' front page links to the other.

The relative sparsity of conservative--liberal blog linkages in this dataset is clearly apparent from Fig.\ \ref{fig:polblogsdata}. Thus it is often used to illustrate the notion of network community structure (see, e.g.,~\cite{newman2011communities}).  At the same time, Fig.\ \ref{fig:polblogsdata} also makes clear that the dataset exhibits additional heterogeneity not fully captured by a simple division of its weblogs into two communities, and indeed recent work also provides evidence of its additional block structure~\cite{krzakala2013spectral}.  Thus the network histogram provides a natural tool to explore the data.

\begin{figure}[t]
\hspace{2.15cm}\includegraphics[width=0.7\columnwidth,trim=50 0 0 50, clip]{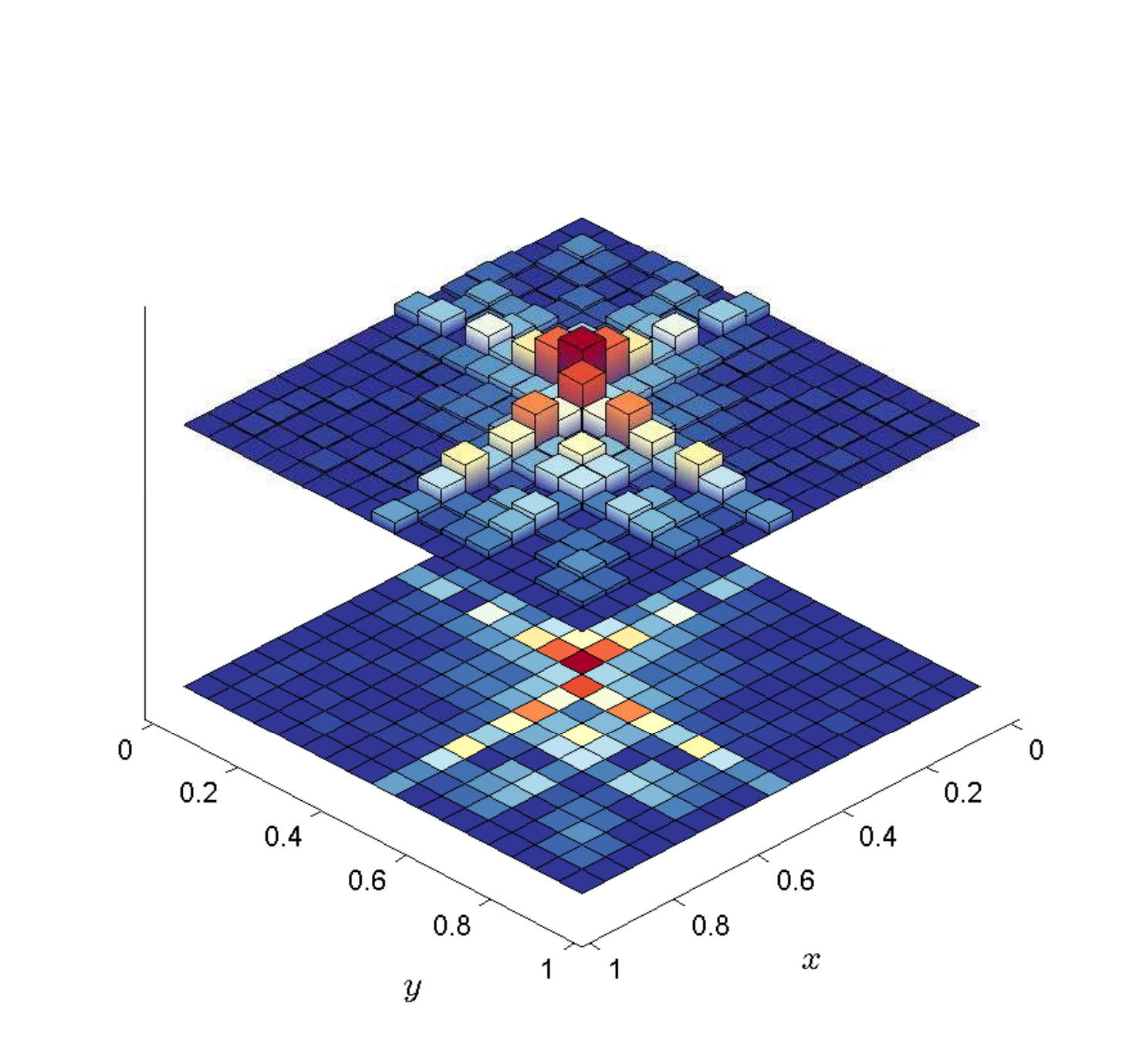}
\vspace{-0.75\baselineskip}%
\caption{\label{fig:polblogshist} Network histogram $\hat f\left(x,y\right)^{\frac{1}{2}}$ fitted to political weblog data. The square root stabilizes the variance of the bin heights and is solely for ease of visualization.}
\end{figure}

Figure~\ref{fig:polblogshist} shows a fitted histogram $\smash{ \hat f\left(x,y\right) } $ obtained from the $n = $ 1224 blogs with at least one link to another blog in the dataset.  From~\eqref{eq:Mhatsqrd-est} we obtained an estimate $ \smash{ \widehat{M^2} } $ in the range 1.1--1.25 for $c$ in the range 3--5, and so the estimated oracle error bound of~\eqref{eq:orac-mise-hstar} evaluates to approximately 1.8 $\times$ 10$^{-2}$.  The bandwidth $ \smash{ \widehat{h^*} } $ was then determined using~\eqref{eq:bw-rule}, and was found to evaluate to 72--74 for $c$ in the range 3--5.  We rounded this to $h = $ 72 to obtain the $ k = $ 17 equal-sized histogram bins that comprise Figs.\ \ref{fig:polblogshist} and \ref{fig:polblogsplits}.  The marginal edge probability estimator $ \hat \rho = \smash{ \sum_{i<j} A_{ij} / \binom{n}{2} } $ evaluates to 16,715 / 748,476 $=$ 2.2332 $\times$ 10$^{-2}$, implying that each off-diagonal histogram bin has approximately 116 effective degrees of freedom.

\begin{figure}[t]
\hspace{1.1cm}\includegraphics[width=0.8\columnwidth]{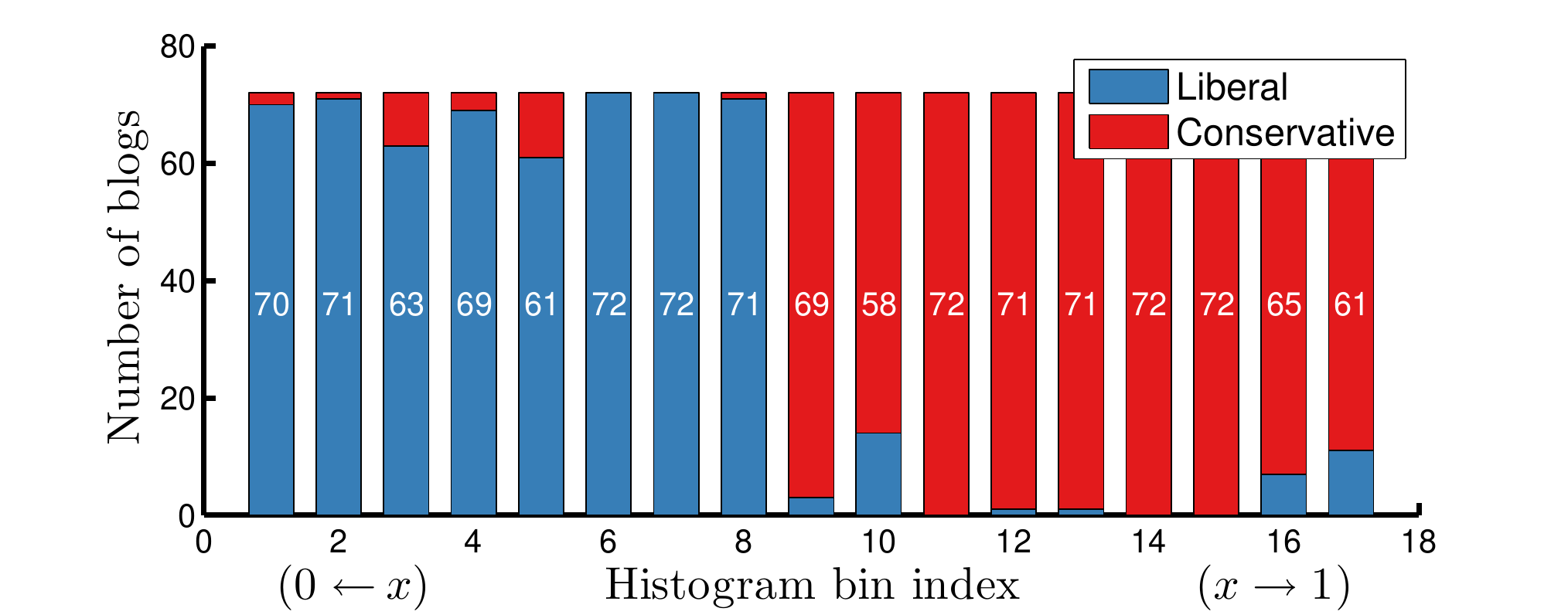}
\caption{\label{fig:polblogsplits} Political affiliation of weblogs within each fitted group, ordered relative to Fig.\ \ref{fig:polblogshist}.  Affiliation counts are shown in white, out of 72 blogs per group.  Nineteen of the 20 most influential liberal blogs identified by~\cite{adamic2005} are assigned to group 8, while 17 of the top 20 conservative blogs are assigned to group 9.}
\end{figure}

\begin{figure*}[t]
\hspace{-2.05cm}\includegraphics[width=1.33\columnwidth]{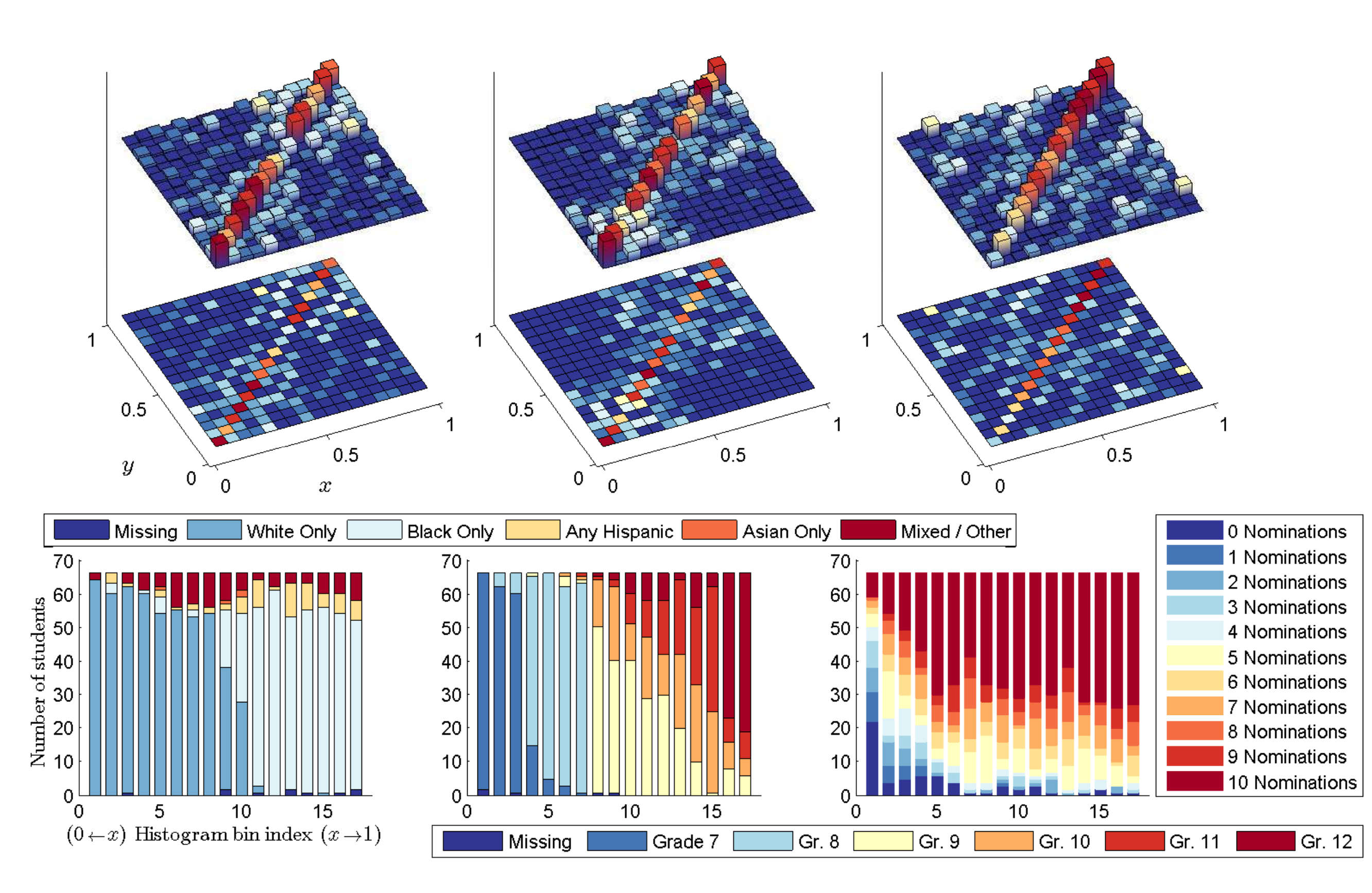}
\caption{\label{fig:addhealthsplit} Network histogram $\hat f\left(x,y\right)^{\frac{1}{2}}$ fitted to student friendship data (top row), with bins ordered according to mean covariate value for race (bottom left), school year (bottom center), and number of friend nominations (bottom right).  The histogram structure visible with respect to each of these three covariates is discussed in the text.}
\end{figure*}

Since exact maximization of the likelihood of~\eqref{eq:MPLE} is known to be computationally infeasible, we obtained the fit shown in Fig.\ \ref{fig:polblogshist} by implementing a simple stochastic search algorithm that swaps pairs and triples of node group memberships selected at random until a local optimum is reached in the likelihood of~\eqref{eq:MPLE}.  The log-likelihood of the data under the fitted model, normalized by the estimated effective degrees of freedom $ \binom{n}{2} \hat \rho $, is $-$2.8728.  To explore as full a range as possible of local likelihood optima, we started from several hundred random configurations, inspected the largest 5\% of returned local maxima, and then repeatedly re-optimized after randomly swapping up to 100 group membership pairs in the best returned solution.

The histogram bin index, relative to the $x$ and $y$ axes of Fig.\ \ref{fig:polblogshist}, allows comparison with the leftmost panel of Fig.\ \ref{fig:polblogsdata}.  Bin indices are arranged first by majority grouping---liberal or conservative---and then by the strength of each fitted group's cross-party connections.  Each node's political affiliation can be viewed as an observed binary covariate that partially explains the network structure.  Below we will consider the more general setting of multiple categorical covariates.

As summarized in Figs.\ \ref{fig:polblogshist} and \ref{fig:polblogsplits}, the coarsest feature of this network is its polarization into sets of dense linkages within the two political blocs of liberal and conservative ideologies.  We also observe from Fig.\ \ref{fig:polblogshist} that nearly 40\% of the histogram bins are empty, in keeping with the sparsity pattern of the data observed in the leftmost panel of Fig.\ \ref{fig:polblogsdata}.  The most densely connected groups of weblogs in both parties show considerable cross-party linkage structure.  This is apparent both from the center region of Fig.\ \ref{fig:polblogshist}, as well as the groupings of Fig.\ \ref{fig:polblogsplits}, in which the most influential blogs identified by~\cite{adamic2005} are seen to be placed in the center of the histogram.  Such features are examples of network microstructure, corresponding to variation at scales smaller than the large fractions of a network that would be captured by a blockmodel with a fixed number of groups.

\subsection{Student friendship data}\label{sec:examples2}

Network datasets often have additional covariates measured at nodes or edges.  To illustrate how to use such information to interpret network histograms, we analyze a student friendship network from the US National Longitudinal Study of Adolescent Health (Add Health)~\cite{resnick1997protecting}. As part of this study, students were asked to identify their gender, race, and school year (grades 7--12), and then to nominate up to 5 friends of each gender.  We consider an undirected version of the resulting network, with a link present whenever either of a pair of students has nominated the other.

We chose to analyze School 44 from the Add Health study, a relative large and racially diverse example among the over 80 schools for which data were collected~\cite{moody2001race}, and one that has been previously analyzed in~\cite{hunter2008goodness} using exponential random graph models. It comprises a main high school with grades 9--12 and a sister ``feeder'' school with grades 7 and 8. We removed 21 zero-degree nodes as well as 5 nodes corresponding to students for which any two of gender, grade, or race covariates were missing, yielding $ n = 1122 $ nodes.

To fit the histogram shown in Fig.\ \ref{fig:addhealthsplit}, we employed the same bandwidth selection procedure and optimization algorithm as above. This yielded a bandwidth $ \smash{ \widehat{h^*} } $ in the range 69--70 for $c$ in the range 3--5, which we rounded down to $h=$ 66 to obtain $k=$ 17 equal-sized histogram bins. This is sparser than the political weblog network considered above, but at the same time $ \smash{ \widehat{M^2} } $ evaluates to 3.2--3.5, indicating relatively less smoothness.  The estimated oracle error bound of~\eqref{eq:orac-mise-hstar} is then approximately 5.6 $\times$ 10$^{-2}$, and our fit yielded a normalized data log-likelihood of $-$4.1714.  For this example, the marginal edge probability estimator $ \hat \rho = \smash{ \sum_{i<j} A_{ij} / \binom{n}{2} } $ evaluates to 5,048 / 628,881 $=$ 8.0270 $\times$ 10$^{-3}$, implying that each off-diagonal histogram bin has approximately 35 effective degrees of freedom.

To explore the fitted groups, we ordered them post-hoc via the mean covariate value per bin for race (coded 0--5), grade (coded 6--12), and number of friends nominated (coded 0--10).  The resulting histograms are shown in the top row of Fig.\ \ref{fig:addhealthsplit}, while the bottom row shows the number of covariate categories comprising each bin.  In the leftmost column of Fig.\ \ref{fig:addhealthsplit}, we observe that the connectivity structure associated with race divides most of the white and black students into two separate groupings, with a decreased tendency to link across these categories.  In the middle column we observe a similar effect for grade, as well as an even stronger effect between the two separate schools: students in grades 7--8 have relatively few interactions with students in grades 9--12.  There is evidence for more mixing within the latter school, with the exception of grade 12, while in the former school the division between grades 7 and 8 is strong.  Finally, in the rightmost column of Fig.\ \ref{fig:addhealthsplit} we see a strong effect associated with the number of friends nominated, which serves as a rough proxy for the degree of each network node.  Diagonal bins in this histogram are ordered almost exclusively from smallest to largest, and we see none of the assortativity associated with race or grade that was so apparent in the previous histogram orderings.

From this example we conclude that the network histogram can provide not only an effective summary of network interactions, but one which is also interpretable in the context of additional covariate information.  This type of aggregate summary allows a fine-grained but concise view of adolescent student friendship networks, and suggests that aggregate statistics on race and grade within a particular school may not be sufficient to give a full picture of the reported social interactions amongst its students.

\section{Discussion}

We argue that the blockmodel is \emph{universal} as a tool for representing interactions in an unlabeled network.  As we use more blocks in our representation, we improve our approximation of the underlying data-generating mechanism, albeit at the cost of increasing complexity.  The results in this article give us insight into how to control the tradeoff between complexity and precision, leading to a flexible nonparametric summary of a network akin to an ordinary histogram.

There is a clear philosophical distinction between the network histogram and the stochastic blockmodel.  The network histogram yields a nonparametric summary of link densities across a network.  In contrast, the stochastic blockmodel was originally conceived as a generative statistical model, meaning that it is typically analyzed in settings where it is presumed to be correctly specified as the data-generating mechanism.  We have instead shown how it can be useful in the case when the blockmodel serves simply to approximate the generating mechanism of the network---a much milder assumption.

To make the network histogram into a useful practical tool, we have derived a procedure for automatically selecting an analysis bandwidth under the assumption of a smooth (H\"older continuous) graphon.  If the graphon has finitely many discontinuities parallel to its $x$- and $y$-axes, for example if it corresponds to an actual blockmodel, then good estimation properties can still be achieved, in analogy to ordinary histogram estimates~\cite{van1985mean}.  In such scenarios the rates at which estimation errors decay are not yet established; indeed, exploring different graphon smoothness classes, and the networks they give rise too, remains an important avenue of future investigation.

As a final point, networks are rarely explored in the absence of other data. A network histogram is defined only up to permutation of its bins, and so to aid in its interpretation we may use other observed variables, labels, or covariates to inform our choice of bin ordering.  As our second data analysis example has shown in the context of student friendship networks, multiple representations can be useful in different ways, and more than one such visual representation can yield insight into the generating mechanism of the network.  In this way the universality of the blockmodel representation is a key piece in the puzzle of general network understanding.  Our results suggest a fundamental re-think of the interpretation of network communities, in light of the fact that many different community assignments can all give an equally valid representation of the network.

\section*{Acknowledgments}

Work supported in part by the US Army Research Office under PECASE Award W911NF-09-1-0555 and MURI Award W911NF-11-1-0036; by the US Office of Naval Research under Award N00014-14-1-0819; by the UK EPSRC under Mathematical Sciences Leadership Fellowship EP/I005250/1, Established Career Fellowship EP/K005413/1 and Developing Leaders Award EP/L001519/1; by the UK Royal Society under a Wolfson Research Merit Award; and by Marie Curie FP7 Integration Grant PCIG12-GA-2012-334622 within the 7th European Union Framework Program.

\appendix

\section{Auxiliary results for the proof of Theorem 1}\label{app}

Throughout we assume that $f$ is a symmetric function on $(0,1)^2$ that is also $\alpha$-H\"older continuous for some $0 < \alpha \leq 1$, with $f \in \operatorname{\textrm{H\"older}}^\alpha(M)$ meaning that
\begin{equation*}
\sup_{(x,y) \neq (x',y') \in (0,1)^2} \frac{ \left| f\left(x,y\right) - f\left(x',y'\right) \right| }{ \left| \left(x,y\right) - \left(x',y'\right) \right|^\alpha } \leq M < \infty ,
\end{equation*}
where $\left|\cdot\right|$ is the Euclidean metric on ${\mathbb{R}}^2$.

We also define a set of summation indices $R_{ab}$, which is the range of values of $i < j$ over which one must aggregate $A_{ij}$ to retrieve $\bar A^*_{ab}$. We write
\begin{align*}
h_a & :=h \I(a<k)+(h+r)\I(a=k);
\\
h^2_{ab} & := \left|R_{ab}\right|=\begin{cases}
h^2 & \mathrm{if} \quad 1\le a<b<k,\\
\tbinom{h}{2} & \mathrm{if} \quad 1\le a=b<k,\\
h \cdot (h+r)& \mathrm{if} \quad 1\le a<b=k,\\
\tbinom{h+r}{2} & \mathrm{if}\quad  a=b=k.\\
\end{cases}
\end{align*}

\begin{proposition}[Moments of $\bar A^*_{ab} $]\label{momlabel2}
Let $f\in\operatorname{\textrm{H\"older}}^\alpha(M)$ be symmetric on $(0,1)^2$, and let the labeling $\tilde z_i$ be determined from the latent vector $\xi$ by
\begin{equation*}
\tilde z_i=\min\left\{\lceil \left(i\right)^{-1}/h \rceil,k\right\},
\end{equation*}
where $(i)^{-1}$ is the rank of $\xi_i$ from smallest to largest. Thus $(i)$ is defined as the index chosen so that $\xi_{(1)}\le \xi_{(2)}\le \dots \le \xi_{(n)}$, and $(i)^{-1}$ is its inverse function.

Assign $i_n=i/(n+1)$ for $i=1,\dots,n$, and define the oracle estimator of $f\left(x,y\right)$ based on knowledge of $\xi$ in terms of the quantities
\begin{align}
\nonumber
\bar A^*_{ab} &=\frac{\sum_{i<j} A_{ij}\I\left(\tilde z_i=a \right)\I\left(\tilde z_j=b \right)}{\sum_{i<j}\I\left(\tilde z_i=a \right)\I\left(\tilde z_j=b \right)}, \quad 1 \leq a, b \leq k.
\end{align}
With these definitions, the means and variances of each oracle estimator component $\bar A^*_{ab} $ satisfy the following:
\begin{align*}
\left| \E \bar A^*_{ab} - \rho_n \bar f_{ab} \right|&\le \rho_n M \left(2 n\right)^{-\alpha/2}\left\{1+o(1)\right\},
\\
\left|\var \bar A^*_{ab} -\frac{\rho_n\bar{f}_{ab}-\rho_n^2\overline{f^2}_{ab}}{h_{ab}^2}\right|&\le
\rho_n \frac{M}{h_{ab}^2 (2n)^{\alpha/2}}\left\{1+o(1)\right\}+\rho_n^2M^2\left(2n\right)^{-\alpha};
\end{align*}
where $\bar{f}_{ab}$ and $\overline{f^2}_{ab}$ are defined by
\begin{align*}
\bar{f}_{ab}&=\frac{1}{\left|\omega_{ab} \right|}\iint_{\omega_{ab}}f(x,y)\,dx\,dy, \quad \overline{f^2}_{ab}=\frac{1}{\left|\omega_{ab} \right|}\iint_{\omega_{ab}}f^2(x,y)\,dx\,dy;
\end{align*}
and the region $\omega_{ab} $ is given by
\begin{equation}
\label{omegaab}
\omega_{ab}=\begin{cases}\left[\left(a-1\right)h/n,a h/n\right]\times \left[\left(b-1\right)h/n,b h/n\right] &\text{if $a<k$ and $b<k$,}\\
\left[\left(k-1\right)h/n, 1\right]\times \left[\left(b-1\right)h/n,b h/n\right] & \text{if $a=k$ and $b<k$,} \\
\left[\left(b-1\right)h/n,b h/n\right]\times \left[\left(k-1\right)h/n, 1\right]  &\text{if $a <k$ and $b=k$,}\\
\left[\left(k-1\right)h/n, 1\right] \times \left[\left(k-1\right)h/n, 1\right] &\text{if $a=k$ and $b=k$.}
\end{cases}
\end{equation}
\end{proposition}

\begin{proof}
Note that the oracle sample proportion estimator takes the form
\begin{align*}
\bar A^*_{ab} &=\frac{\sum_{i<j} A_{ij}\I\left(\tilde z_i=a \right)\I\left(\tilde z_j=b \right)}{\sum_{i<j}\I\left(\tilde z_i=a \right)\I\left(\tilde z_j=b \right)}\\
&=\begin{cases}
\frac{\textstyle\sum_{ j =  h (b-1)+1 }^{  h b } \sum_{ i =  h (a-1)+1}^{  h a \I\left(a\neq b\right)+(j-1)\I\left(a=b\right)}
A_{(i)(j)}}{  h_{ab}^2} & \text{if $a<k$ and $b<k$,} \\
\frac{\textstyle\sum_{ j =  h (k-1)+1 }^{  n} \sum_{ i =  h (a-1)+1}^{  h a \I\left(a\neq b\right)+(j-1)\I\left(a=b\right)}
A_{(i)(j)}}{  h_{a k}^2} & \text{if $a \leq k$ and $b=k$,} \\
\bar A^*_{b k} & \text{if $a = k$ and $b \leq k$;}
\end{cases}\\
&=\frac{\sum_{(i,j)\in R_{ab}} A_{(i)(j)}}{h_{ab}^2},
\end{align*}
where $R_{ab}$ is defined implicitly to make the summation valid, and is non-random.
Thus we may conclude that
\begin{equation}
\label{Ebar1}
\E \bar A^*_{ab}
= \frac{ 1 }{ h_{ab}^2 } \sum_{(i,j)\in R_{ab}} \E A_{ (i) (j) } .
\end{equation}
We define $ \tilde f_{ a b } $, for $i_n=i/(n+1)$ and $j_n=j/(n+1)$, as
\begin{equation}
\label{ftilde}
\tilde f_{ a b } = \frac{ 1 }{ h_{ab}^2 } \sum_{ (i,j)\in R_{ab}} f\left(i_n,
j_n\right) .
\end{equation}
We then use~\eqref{exp1} from Lemma~\ref{momlabel} to obtain that
\begin{align}
\label{exp2}
\left|\E  \bar A^*_{ ab  } -\rho_n \tilde f_{ a b }\right|\le \rho_n M\left\{2(n+2) \right\}^{-\alpha/2}.
\end{align}
We note from Lemma~\ref{quadmomlabel1}
that as $f\in\operatorname{\textrm{H\"older}}^\alpha(M)$ on $(0,1)^2$,
\begin{equation}
\label{ftildebar}
|\tilde{f}_{ab}-\bar{f}_{ab} |<M \, 2^{\alpha/2}n^{-\alpha}\left\{1+2^\alpha \I\left(a=b\right)\right\}.
\end{equation}
We then apply the triangle inequality to~\eqref{Ebar1}--\eqref{ftildebar} to derive
\begin{align*}
\left| \E \bar A^*_{ab} - \rho_n \bar f_{ab} \right|
 & \le \rho_n M \left[\left\{ 2(n+2) \right\}^{ - \alpha / 2 }+2^{\alpha/2}n^{-\alpha}\left\{1+2^\alpha \I\left(a=b\right)\right\}\right]\\
&\le \rho_n M (2n)^{-\alpha/2}\left\{1+o(1)\right\}.
\end{align*}

This establishes the form of  $\E \bar A^*_{ab} $.  We next calculate
\begin{align*}
\var \bar A^*_{ab} =\frac{\sum_{(i,j)\in R_{ab}} \sum_{(m,l)\in R_{ab}} \cov\left\{ A_{(i)(j)},A_{(m)(l)}\right\}}{h_{ab}^4}.
\end{align*}
Referring to~\eqref{var1} of Lemma~\ref{momlabel},
\begin{align*}
\var \bar A^*_{ab} &=\frac{1}{h_{ab}^4}\sum_{(i,j)\in R_{ab}} \sum_{(m,l)\in R_{ab}} \cov\{A_{ (i) (j) },A_{ (m) (l) }\}
\\ & \le \frac{1}{h_{ab}^4}\sum_{(i,j)\in R_{ab}}
\rho_n f\left(i_n,j_n\right)\left\{1-\rho_n f\left(i_n,j_n\right)\right\} +
\rho_n^2M^2\left[2(n+2)\right]^{-\alpha} \\
&\qquad+\frac{\rho_n}{h_{ab}^2}M \left\{2(n+2) \right\}^{-\alpha/2}\left[1+\rho_n M \left\{2(n+2) \right\}^{-\alpha/2}\right].
\end{align*}
We may likewise determine the lower bound of
\begin{align*}
\var \bar A^*_{ab} & \ge \frac{1}{h_{ab}^4}\sum_{(i,j)\in R_{ab}}
\rho_n f\left(i_n,j_n\right)\left\{1-\rho_n f\left(i_n,j_n\right)\right\} -
\rho_n^2M^2\left[2(n+2)\right]^{-\alpha} \\
&\qquad -\frac{\rho_n}{h_{ab}^2}M \left\{2(n+2) \right\}^{-\alpha/2}\left[1+\rho_n M \left\{2(n+2) \right\}^{-\alpha/2}\right].
\end{align*}
From Lemmas~\ref{quadmomlabel1} and~\ref{quadmomlabel} below, writing $\overline{f^2}_{ab}$ for the normalized integral of $f^2(x,y)$ over the block $\omega_{ab}$, we have respectively that
\begin{equation*}
\left|\tilde f_{ a b }-\bar{f}_{ab}\right|\le
\frac{M2^{\alpha/2}}{n^\alpha}\left\{1+2^\alpha\I\left(a=b\right)\right\}
\end{equation*}
and
\begin{equation*}
\left|\frac{1}{h_{ab}^2}\sum_{(i,j)\in R_{ab}} f^2\left(i_n,j_n\right)-\overline{f^2}_{ab}\right|\le \frac{2\|f\|_{\infty} M 2^{\alpha/2}}{n^{\alpha}}\left\{1+2^\alpha\I(a=b)
\right\}.
\end{equation*}
Together these results yield the claimed expression for the variance of $\bar A^*_{ab}$.
\end{proof}

\begin{lemma}\label{normbound2}
Let $ f \in\operatorname{\textrm{H\"older}}^\alpha(M)$, with $\bar f_{ab}=\left|\omega_{ab} \right|^{-1}\iint_{\omega_{ab}}f(x,y)\,dx\,dy$ defined as its local average over $\omega_{ab}$. Then
\begin{equation*}
\frac{1}{\left| \omega_{ab}\right|}\iint_{\omega_{ab}} \left| f\left(x,y\right) - \bar f_{ab} \right|^2\,dx\,dy\le
M^2 2^{\alpha} (h/n)^{2\alpha}\left\{1+2^{2\alpha}\I\left(a=k\; {\mathrm{or}}\;b=k\right)\right\}.
\end{equation*}
\end{lemma}

\begin{proof}
Recall that $ \omega_{ab} $ is given by~\eqref{omegaab},
as before.  Note from the definition of the set $\operatorname{\textrm{H\"older}}^\alpha(M)$ that
if $(x,y)\in \omega_{ab} $  and $a,b<k$, then
\begin{align*}
\left| \bar f_{ab} - f\left(x,y\right) \right| &= \left| \frac{1}{\left| \omega_{ab}\right|}\iint_{\omega_{ab}}
f\left(x',y'\right)\,dx'\,dy'-f(x,y) \right|\\
\Rightarrow
\left| \bar f_{ab} - f\left(x,y\right) \right| &\le  \frac{1}{\left| \omega_{ab}\right|}\iint_{\omega_{ab}}
\left|f\left(x',y'\right)-f(x,y) \right|\,dx'\,dy'\\
&\le \frac{1}{\left| \omega_{ab}\right|}\iint_{\omega_{ab}}
 M \left|\left(x',y'\right)-(x,y) \right|^{\alpha}\,dx'\,dy'\\
&\le  \frac{1}{\left| \omega_{ab}\right|}\iint_{\omega_{ab}}
 M \left[2 (h/n)^2 \right]^{\alpha/2}\,dx'\,dy'=M 2^{\alpha/2} (h/n)^{\alpha}.\end{align*}
Thus
\begin{align*}
\left| \bar f_{ab} - f\left(x,y\right) \right|^2 &\le  M^2 2^{\alpha} (h/n)^{2\alpha}\\
\Rightarrow
 \frac{1}{\left| \omega_{ab}\right|}\iint_{\omega_{ab}} \left| f\left(x,y\right) - \bar f_{ab} \right|^2\,dx\,dy& \le M^2 2^{\alpha} (h/n)^{2\alpha}.
\end{align*}
If $a=k$ or $b=k$ then we replace $h$ by $2h$ to obtain a bound.
\end{proof}
Lemma~\ref{normbound2} has been adapted from Wolfe and Olhede~[16].

\begin{lemma}[Moments of $A_{(i)(j)}$]\label{momlabel}
Let $i_n=i/(n+1)$ for $i=1,\dots,n$, and let $(i)$ be defined as the index chosen so that $\xi_{(1)}\le \xi_{(2)}\le \dots \le \xi_{(n)}$.  Then the means and variances of each $A_{ (i) (j) }$ for $i<j$ satisfy the following:
\begin{align}
\label{exp1}
\left|\E A_{ (i) (j) }-\rho_n f\left(i_n,j_n \right)\right| & \le \rho_n M\left\{2(n+2) \right\}^{-\alpha/2},\\
\label{var1}
\left|\var A_{ (i) (j) }-\rho_n f\left(i_n,j_n\right)\left(1-\rho_n f\left(i_n,j_n\right)\right) \right| & \leq \rho_n
\cdot M \left\{2(n+2) \right\}^{-\alpha/2}\\
\nonumber
&\qquad \cdot \left[1+\rho_n M \left\{2(n+2) \right\}^{-\alpha/2}\right]\!.\\
\nonumber
\text{For $i\neq m$  or $j\neq l$,} \qquad \cov\left\{A_{ (i) (j) } ,A_{(m) (l)}\right\} & \le \rho_n^2M^2\left\{2(n+2)\right\}^{-\alpha}.
\end{align}
\end{lemma}

\begin{proof}
Equation~\eqref{exp1} follows directly from the law of iterated expectation, with the first calculation following from conditioning on $\xi$:
\begin{align}
\E A_{ (i) (j) }&=\E_\xi\left[\E_{A|\xi}\left\{ A_{ (i) (j) }\,\vert\,\xi\right\}\right]=\E_\xi\left\{\rho_n f\left(\xi_{(i)},\xi_{(j)}\right)\right\},
\label{Aij}
\end{align}
and the second calculation following by approximation of the latter expectation, as we now show. As $|\cdot |$ is convex, Jensen's inequality permits us to deduce that
\begin{equation}
\label{exp-delta}
\left|\E_\xi \rho_n f\left(\xi_{(i)},\xi_{(j)}\right)-\rho_nf\left(i_n,j_n\right)\right|\le \rho_n\E_\xi\left\{ \left|f\left(\xi_{(i)},\xi_{(j)}\right)-f\left(i_n,j_n\right)\right|\right\}.
\end{equation}
We note that from Lemma~\ref{expdev}, we have
\begin{align}
\label{expxi}
\E_\xi \left|f\left(\xi_{(i)},\xi_{(j)}\right)-f\left(i_n,j_n\right)\right| \leq M \left\{ 2(n+2) \right\}^{ - \alpha / 2 },
\end{align}
and so we can deduce~\eqref{exp1} by combining~\eqref{Aij}--\eqref{expxi}.

Equation~\eqref{var1} is derived from the law of total variance by
\begin{align}
\nonumber
\var A_{ (i) (j) }&=\E_{\xi}\left[\var_{A|\xi} \left\{A_{ (i) (j) } \right\}\right]+
\var_{\xi}\left[\E_{A|\xi} \left\{A_{ (i) (j) } \right\}\right]
\\
\label{varterms}
&=\E_{\xi}\left\{\rho_n f\left(\xi_{(i)},\xi_{(j)} \right)\left(1- \rho_n f\left(\xi_{(i)},\xi_{(j)} \right)\right) \right\}\\
&\qquad +\E_{\xi}\left\{\rho_n^2 f^2\left(\xi_{(i)},\xi_{(j)} \right)\right\}-\E_{\xi}^2\left\{\rho_n f\left(\xi_{(i)},\xi_{(j)} \right)\right\},\; 1\le  i<j\le n.
\nonumber
\end{align}
The second and third terms in~\eqref{varterms} cancel, and thus we obtain that
\begin{align*}
\var A_{ (i) (j) }&=\rho_n \left\{ \E_{\xi} f\left(\xi_{(i)},\xi_{(j)} \right) \right\} \left\{ 1-\rho_n \E_{\xi} f\left(\xi_{(i)},\xi_{(j)} \right) \right\} .
\end{align*}
We now need to calculate expectations with respect to the latent vector $\xi$. Owing to~\eqref{expxi},
we can upper bound $\E_{\xi} f\left(\xi_{(i)},\xi_{(j)} \right) $ by the quantity $\rho_n f\left(i_n,j_n \right)+ \rho_n M\left\{2(n+2) \right\}^{-\alpha/2}$, and likewise the negative term $-\E_{\xi} f\left(\xi_{(i)},\xi_{(j)} \right)$ by the quantity $-\rho_n f\left(i_n,j_n \right)+ \rho_n M\left\{2(n+2) \right\}^{-\alpha/2}$. Similarly, $1 - \rho_n \E_{\xi} f\left(\xi_{(i)},\xi_{(j)} \right)$ and its negative can be lower bounded. Thus we may deduce the two inequalities
\begin{align*}
\nonumber
\var A_{ (i) (j) }&\le \rho_n\left[f\left(i_n,j_n\right)+M \left\{ 2(n+2) \right\}^{ - \alpha / 2 }\right]\\
\nonumber
&\qquad \cdot \left[1
-\rho_n f\left(i_n,j_n\right)+\rho_n M \left\{ 2(n+2) \right\}^{ - \alpha / 2 } \right],\\
\var A_{ (i) (j) }&\ge\rho_n\left[f\left(i_n,j_n\right)-M \left\{ 2(n+2) \right\}^{ - \alpha / 2 }\right]\\
\nonumber
&\qquad \cdot \left[1
-\rho_n f\left(i_n,j_n\right)-\rho_n M \left\{ 2(n+2) \right\}^{ - \alpha / 2 } \right].
\end{align*}
Combining these two relationships, we obtain~\eqref{var1}.

From the law of total covariance, we have that since $i<j$ and $m<l$, when at least either $i\neq m$ or $j \neq l$, the conditional independence of the Bernoulli trials comprising $A$ yields
\begin{align}
\label{covary1}
\cov\left\{A_{ (i) (j) } ,A_{(m) (l)}\right\}&=\E_{\xi}\left[\cov_{A|\xi} \left\{A_{ (i) (j) } ,A_{ (m) (l) } \right\}\right]\\
\nonumber
&\qquad+
\cov_{\xi}\left[\E_{A|\xi} \left\{A_{ (i) (j) } \right\},\E_{A|\xi} \left\{A_{ (m) (l) } \right\}\right]\\
&=\rho_n^2 \cov_\xi \left\{f\left(\xi_{(i)},\xi_{(j)} \right),f\left(\xi_{(m)},\xi_{(l)} \right)\right\}.
\nonumber
\end{align}
We now simplify this expression further, working directly with the form in~\eqref{covary1}. We define $j_n=j/(n+1)$, as well as $m_n=m/(n+1)$ and $l_n=l/(n+1)$. We then use the shift-invariance of the covariance operator to write
\begin{multline*}
\left|\cov_\xi\left\{f\left(\xi_{(i)},\xi_{(j)} \right),f\left(\xi_{(m)},\xi_{(l)} \right)\right\}\right|
\\ \le \left|\E_\xi \left[\left\{f\left(\xi_{(i)},\xi_{(j)} \right)-f\left(i_n,j_n\right)\right\}
\left\{f\left(\xi_{(m)},\xi_{(l)} \right)-f\left(m_n,l_n\right)\right\}\right]\right|,
\end{multline*}
where we have a bound, rather than equality, because we do \emph{not} claim that $\E\left\{f\left(\xi_{(i)},\xi_{(j)} \right)\right\}=f\left(i_n,j_n\right)$.
We may use Jensen's inequality to deduce that
\begin{multline*}
\left|\E_\xi \left\{f\left(\xi_{(i)},\xi_{(j)} \right)-f\left(i_n,j_n\right)\right\}\left\{f\left(\xi_{(m)},\xi_{(l)} \right)-f\left(m_n,l_n\right)\right\}\right|
\\ \le \E_\xi \left|\left\{f\left(\xi_{(i)},\xi_{(j)} \right)-f\left(i_n,j_n\right)\right\}\left\{f\left(\xi_{(m)},\xi_{(l)} \right)-f\left(m_n,l_n\right)\right\}\right|.
\end{multline*}
Now, because $f \in \operatorname{\textrm{H\"older}}^\alpha(M)$ by hypothesis, there exists $M<\infty$ such that
\begin{equation*}
\left|f\left(x,y\right)-f\left(x',y'\right)\right|\le M\left|(x,y)-(x',y')\right|^{\alpha} ,
\end{equation*}
and so we obtain that
\begin{multline*}
\E_\xi \left|f\left(\xi_{(i)},\xi_{(j)} \right)-f\left(i_n,j_n\right)\right|\left|f\left(\xi_{(m)},\xi_{(l)} \right)-f\left(m_n,l_n\right)\right|
\\
\le M^2 \E_\xi \left|\left(\xi_{(i)},\xi_{(j)} \right)-\left(i_n,j_n\right) \right|^{\alpha}
\left|\left(\xi_{(m)},\xi_{(l)} \right)-\left(m_n,l_n\right) \right|^{\alpha}.
\end{multline*}
From the Cauchy--Schwarz inequality, it therefore follows that
\begin{multline*}
\E_\xi \left|\left(\xi_{(i)},\xi_{(j)} \right)-\left(i_n,j_n\right) \right|^{\alpha}
\left|\left(\xi_{(m)},\xi_{(l)} \right)-\left(m_n,l_n\right) \right|^{\alpha}\\
 \le \sqrt{ \E_\xi \left|\left(\xi_{(i)},\xi_{(j)} \right)-\left(i_n,j_n\right) \right|^{2\alpha} }
\sqrt{ \E_\xi \left|\left(\xi_{(m)},\xi_{(l)} \right)-\left(m_n,l_n\right) \right|^{2\alpha} } .
\end{multline*}
We then calculate
\begin{align*}
\E_\xi \left|\left(\xi_{(m)},\xi_{(l)} \right)-\left(m_n,l_n\right) \right|^{2\alpha}
&=\E_\xi \left\{ \left(\xi_{(m)}-m_n\right)^2+\left(\xi_{(l)}-l_n\right)^2
\right\}^{\alpha}.
\end{align*}
Applying Jensen's inequality, we find that for $\alpha\leq 1$,
\begin{align*}
\E_\xi \left\{ \left(\xi_{(m)}-m_n\right)^2+\left(\xi_{(l)}-l_n\right)^2
\right\}^{\alpha}&\le \left[ \var \{ \xi_{(l)}\}+ \var \{ \xi_{(j)}\}\right]^{\alpha}\le
\left\{2 (n+2)\right\}^{-2\alpha}.
\end{align*}
Thus we may deduce that
\begin{equation*}
\left|\cov_\xi \left\{f\left(\xi_{(i)},\xi_{(j)} \right),f\left(\xi_{(m)},\xi_{(l)} \right)\right\}\right|
\le M^2
 \left[\left\{2(n+2) \right\}^{-\alpha}
\left\{2(n+2) \right\}^{-\alpha}\right]^{1/2}.
\end{equation*}
Combining this expression with~\eqref{covary1} then yields the stated result.
\end{proof}

\begin{lemma}\label{expdev}
Let $ f \in\operatorname{\textrm{H\"older}}^\alpha(M)$, and let $\{\xi_{(i)}\}_{i=1}^n$ be an ordered sample of independent $\operatorname{Uniform}(0,1)$ random variables. Then for $ 1 \leq i, j \leq n $ we have
\begin{equation*}
\E_\xi \, \bigl| f\left( \xi_{(i)},\xi_{(j)}\right) - f\bigl( i_n , j_n \bigr) \bigr|
\leq M \left\{ 2(n+2) \right\}^{ - \alpha / 2 } .
\end{equation*}
\end{lemma}

\begin{proof}
We note that as $ f \in\operatorname{\textrm{H\"older}}^\alpha(M)$,
\begin{equation*}
\label{eq:lips-ij}
\left| f\left( \xi_{(i)},\xi_{(j)}\right)-f\left( i_n,j_n\right)\right|\leq M\left| (
\xi_{(i)} , \xi_{(j)})-(
i_n, j_n)\right|^{\alpha}, \quad 1 \leq i, j \leq n .
\end{equation*}
Since $ \var \xi_{(i)} = i_n ( 1 - i_n ) / (n+2) \leq (1/4) / (n+2)$,
by Jensen's inequality we have for any $0 < \alpha \leq 1 $ that
\begin{equation*}
\E_{\xi} \left\{ ( \xi_{(i)} - i_n )^2 + ( \xi_{(j)} - j_n )^2 \right\}^{ \alpha/ 2 }
 \le \left( \var \xi_{(i)}+\var \xi_{(j)} \right)^{ \alpha / 2 }
\leq \left\{ 2(n+2) \right\}^{ - \alpha / 2 }.
\end{equation*}
This completes the proof.
\end{proof}

Lemma~\ref{expdev} has been adapted from Wolfe and Olhede~[16].

\begin{lemma}[Linear quadrature bounds]\label{quadmomlabel1}
Let $ f \in\operatorname{\textrm{H\"older}}^\alpha(M)$ be a symmetric function on $(0,1)^2$, and define $i_n=i/(n+1)$, $j_n=j/(n+1)$.  Then with
\begin{equation*}
\tilde f_{ a b } = \frac{ 1 }{ h_{ab}^2 } \sum_{ (i,j)\in R_{ab}} f\left(i_n,
j_n\right) ,\quad 1\le a\le b\le k,
\end{equation*}
we have that
\begin{equation*}
\left|\tilde f_{ a b }-\bar{f}_{ab}\right|\le
M \, 2^{\alpha/2}n^{-\alpha}\left\{1+2^\alpha \I\left(a=b\right)\right\}.
\end{equation*}
\end{lemma}

\begin{proof}
We start from the definition of
\begin{equation*}
\nonumber
 \tilde f_{ a b }= \frac{1}{h_{ab}^{2 }} \sum_{(i,j)\in R_{ab}} f\left(i_n,j_n\right).
\end{equation*}
Thus we may by simple expansion determine
\begin{align*}
\nonumber
\tilde f_{ a b }&=\frac{n^2}{h_{ab}^2}
\sum_{(i,j)\in R_{ab}}
\int_{\frac{j-1}{n}}^{\frac{j}{n}}\int_{\frac{i-1}{n}}^{\frac{i}{n}}\left[f\left(x,y\right)+f\left(i_n,j_n\right)-f\left(x,y\right) \right]\,dx\,dy\\
&=\frac{n^2}{h_{ab}^2} \sum_{(i,j)\in R_{ab}}\left[
\int_{\frac{j-1}{n}}^{\frac{j}{n}}\int_{\frac{i-1}{n}}^{\frac{i}{n}}f\left(x,y\right)\,dx\,dy+\int_{\frac{j-1}{n}}^{\frac{j}{n}}\int_{\frac{i-1}{n}}^{\frac{i}{n}}\left\{f\left(i_n,j_n\right)-f\left(x,y\right) \right\}\,dx\,dy.
\right]\end{align*}
We now use the fact that $f\left(x,y\right)\in \operatorname{\textrm{H\"older}}^\alpha(M)$.  Thus we may write
\begin{align}
\nonumber
&\left|\frac{n^2}{h_{ab}^2} \sum_{(i,j)\in R_{ab}}
\int_{\frac{j-1}{n}}^{\frac{j}{n}}\int_{\frac{i-1}{n}}^{\frac{i}{n}}\left\{f\left(i_n,j_n\right)-f\left(x,y\right) \right\}\,dx\,dy\right|\\
&\qquad \le \frac{n^2}{h_{ab}^2} \sum_{(i,j)\in R_{ab}}
\int_{\frac{j-1}{n}}^{\frac{j}{n}}\int_{\frac{i-1}{n}}^{\frac{i}{n}}\left|f\left(i_n,j_n\right)-f\left(x,y\right) \right|\,dx\,dy \le \frac{M2^{\alpha/2}}{n^\alpha},
\label{both}
\end{align}
with the  last inequality following from the fact that $f$ is an $\alpha$-H\"older function on the domain of integration.  Furthermore, we note directly if $a<b$ then, with $\omega_{ab}$ as defined in~\eqref{omegaab},
\begin{align}
\frac{n^2}{h_{ab}^2} \sum_{(i,j)\in R_{ab}}
\int_{\frac{j-1}{n}}^{\frac{j}{n}}\int_{\frac{i-1}{n}}^{\frac{i}{n}}f\left(x,y\right)\,dx\,dy
&=\frac{n^2}{h_{ab}^2}
\iint_{\omega_{ab}}f\left(x,y\right)\,dx\,dy.
\label{offd}
\end{align}
If on the other hand $a=b$ then
\begin{multline}
\frac{n^2}{\tbinom{h_b}{2}} \sum_{(i,j)\in R_{bb}}
\int_{\frac{j-1}{n}}^{\frac{j}{n}}\int_{\frac{i-1}{n}}^{\frac{i}{n}}f\left(x,y\right)\,dx\,dy\\
=\frac{n^2}{\tbinom{h_b}{2}}\sum_{j=(b-1)h+1}^{hb \I(b<k)+n \I(b=k)}\sum_{i=(b-1)h+1}^{j-1}
\int_{\frac{j-1}{n}}^{\frac{j}{n}}\int_{\frac{i-1}{n}}^{\frac{i}{n}}f\left(x,y\right)\,dx\,dy.
\label{ond}
\end{multline}
This equation acknowledges that group $a$ has size $h_a$, which is equal to $h$ for $a=1,\dots,k-1$,
and $h_k=h+r$ for $a=k$.
We shall start by simplifying this expression. We note that the latter becomes:
\begin{align}
\nonumber
\frac{n^2}{\tbinom{h_b}{2}}&\sum_{j=(b-1)h+1}^{hb \I(b<k)+n \I(b=k)}
\int_{\frac{j-1}{n}}^{\frac{j}{n}}\int_{\frac{(b-1)h}{n}}^{\frac{j-1}{n}}f\left(x,y\right)\,dx\,dy
\\
\nonumber
&=\bar f_{bb}+ \!\!\!\! \sum_{j=(b-1)h+1}^{hb \I(b<k)+n \I(b=k)}\int_{\frac{j-1}{n}}^{\frac{j}{n}}\left\{\left[ \frac{n^2}{\tbinom{h_b}{2}}-\frac{2n^2}{h^2_b}\right]\int_{\frac{(b-1)h}{n}}^{y}-\frac{n^2}{\tbinom{h_b}{2}}\int_{\frac{j-1}{n}}^{y}\right\}
f(x,y)\,dx\,dy\\
\nonumber
&=\bar f_{bb}+\frac{2n^2}{h_b}\sum_{j=(b-1)h+1}^{hb \I(b<k)+n \I(b=k)}
\int_{\frac{j-1}{n}}^{\frac{j}{n}}\left\{\left[ \frac{1}{h_b-1}-\frac{1}{h_b}\right]\int_{\frac{(b-1)h}{n}}^{y}-\frac{1}{h_b-1}\int_{\frac{j-1}{n}}^{y}\right\}\\
& \nonumber \hskip5.5cm\cdot
f(x,y)\,dx\,dy\\
\nonumber
&=\bar f_{bb}+\frac{2n^2}{h_b}\sum_{j=(b-1)h+1}^{hb \I(b<k)+n \I(b=k)}
\int_{\frac{j-1}{n}}^{\frac{j}{n}}\left\{\frac{1}{h_b(h_b-1)}\int_{\frac{(b-1)h}{n}}^{y}-\frac{1}{h_b-1}\int_{\frac{j-1}{n}}^{y}\right\}\\
& \nonumber \hskip5.5cm\cdot
f(x,y)\,dx\,dy\\
\nonumber
&=\bar f_{bb}+\frac{1}{(h_b-1)}\sum_{j=(b-1)h+1}^{hb \I(b<k)+n \I(b=k)}
\int_{\frac{j-1}{n}}^{\frac{j}{n}}\left\{\frac{2n^2}{h_b^2}\int_{\frac{(b-1)h}{n}}^{y}-\frac{2n^2}{h_b}\int_{\frac{j-1}{n}}^{y}\right\}
f(x,y)\,dx\,dy\\
&=\bar f_{bb}+\frac{1}{(h_b-1)}\left\{ \bar f_{bb}-\frac{2n^2}{h_b}\sum_{j=(b-1)h+1}^{hb \I(b<k)+n \I(b=k)}
\int_{\frac{j-1}{n}}^{\frac{j}{n}}\int_{\frac{j-1}{n}}^{y}
f(x,y)\,dx\,dy\right\}\nonumber\\
&=\bar f_{bb}+\frac{1}{(h_b-1)}\left\{ \frac{2n^2}{h_b}\sum_{j=(b-1)h+1}^{hb \I(b<k)+n \I(b=k)}
\int_{\frac{j-1}{n}}^{\frac{j}{n}}\int_{\frac{j-1}{n}}^{y}
\left(\bar f_{bb}-f(x,y)\right)\,dx\,dy\right\}.
\label{ond2}
\end{align}
We note that
\begin{multline}
\left|\frac{2n^2}{h_b}\sum_{j=(b-1)h+1}^{hb \I(b<k)+n \I(b=k)}
\int_{\frac{j-1}{n}}^{\frac{j}{n}}\int_{\frac{j-1}{n}}^{y}
\left(\bar f_{bb}-f(x,y)\right)\,dx\,dy\right|\\
\quad \le \frac{2n^2}{h_b}\sum_{j=(b-1)h+1}^{hb \I(b<k)+n \I(b=k)}
\int_{\frac{j-1}{n}}^{\frac{j}{n}}\int_{\frac{j-1}{n}}^{y}
\left|\bar f_{bb}-f(x,y)\right|\,dx\,dy\le M\left(\sqrt{2} h_b/n\right)^{\alpha}.
\label{finalt}
\end{multline}
Thus, combining~\eqref{finalt} with~\eqref{both},~\eqref{ond}, and~\eqref{ond2}, we have
\begin{align*}
\left| \bar f_{bb}-\tilde f_{bb}
\right|\le \frac{M2^{\alpha/2}}{n^\alpha}+M\left(\sqrt{2} h_b/n\right)^{\alpha}\frac{1}{h_b-1}.
\end{align*}
From the off-diagonal entries $a<b$ we may conclude from~\eqref{both} and~\eqref{offd}
that
\begin{align*}
\left| \bar f_{ab}-\tilde f_{ab}
\right|\le \frac{M2^{\alpha/2}}{n^\alpha}.
\end{align*}
Thus it follows that
\begin{equation*}
 \left| \bar f_{ab}-\tilde f_{ab}
\right| \le \frac{M 2^{\alpha/2}}{n^\alpha}
+\begin{cases}
0 & a\neq b,\\
M\left(\sqrt{2} h_b/n\right)^{\alpha}\frac{1}{h_b-1} &a= b.
\end{cases}
\end{equation*}
Since $\frac{h_b^{\alpha}}{(h_b-1)}\le 2^\alpha$ if $h_b\ge 2$, the expression follows.
This concludes the proof.
\end{proof}

Lemma~\ref{quadmomlabel1} has been adapted from Wolfe and Olhede~[16].

\begin{lemma}[Square quadrature bounds]\label{quadmomlabel}
Let $ f \in\operatorname{\textrm{H\"older}}^\alpha(M)$ be a symmetric function on $(0,1)^2$, and define $i_n=i/(n+1)$, $j_n=j/(n+1)$.  Then with
\begin{equation*}
\widetilde{ f^2} _{ a b }= \frac{ 1 }{ h_{ab}^2 } \sum_{ (i,j)\in R_{ab}} f^2\left(i_n,
j_n\right) ,\quad 1\le a\le b\le k,
\end{equation*}
we have that
\begin{equation*}
\left|\widetilde{ f^2} _{ a b }-\overline{f^2}_{ab}\right|\le
\frac{2\|f\|_{\infty} M 2^{\alpha/2}}{n^{\alpha}}\left\{1+2^\alpha\I(a=b)
\right\}.
\end{equation*}
\end{lemma}

\begin{proof}
We start from
\begin{align*}
\nonumber
\widetilde{ f^2} _{ a b }& =\frac{n^2}{h_{ab}^2}
\sum_{(i,j)\in R_{ab}}
\int_{\frac{j-1}{n}}^{\frac{j}{n}}\int_{\frac{i-1}{n}}^{\frac{i}{n}}\left[f^2\left(x,y\right)+f^2\left(i_n,j_n\right)-f^2\left(x,y\right) \right]\,dx\,dy\\
&\quad =\frac{n^2}{h_{ab}^2} \sum_{(i,j)\in R_{ab}}
\left\{\int_{\frac{j-1}{n}}^{\frac{j}{n}}\int_{\frac{i-1}{n}}^{\frac{i}{n}}f^2\left(x,y\right)+\int_{\frac{j-1}{n}}^{\frac{j}{n}}\int_{\frac{i-1}{n}}^{\frac{i}{n}}f^2\left(i_n,j_n\right)-f^2\left(x,y\right) \right\}\,dx\,dy.
\end{align*}
We now use  that $f\in\operatorname{\textrm{H\"older}}^\alpha(M)$.  We write
\begin{align}
\nonumber
&\left|\frac{n^2}{h_{ab}^2} \sum_{(i,j)\in R_{ab}}
\int_{\frac{j-1}{n}}^{\frac{j}{n}}\int_{\frac{i-1}{n}}^{\frac{i}{n}}\left\{f^2\left(i_n,j_n\right)-f^2\left(x,y\right) \right\}\,dx\,dy\right|\\
\nonumber
&\qquad \le \frac{n^2}{h_{ab}^2} \sum_{(i,j)\in R_{ab}}
\int_{\frac{j-1}{n}}^{\frac{j}{n}}\int_{\frac{i-1}{n}}^{\frac{i}{n}}\left|f^2\left(i_n,j_n\right)-f^2\left(x,y\right) \right|\,dx\,dy\\
\nonumber
&\qquad =\frac{n^2}{h_{ab}^2} \sum_{(i,j)\in R_{ab}}
\int_{\frac{j-1}{n}}^{\frac{j}{n}}\int_{\frac{i-1}{n}}^{\frac{i}{n}}\left|f\left(i_n,j_n\right)+f\left(x,y\right) \right|
\left|f\left(i_n,j_n\right)-f\left(x,y\right) \right|\,dx\,dy\\
\nonumber
&\qquad \le \frac{2 \|f\|_{\infty} n^2}{h_{ab}^2} \sum_{(i,j)\in R_{ab}}
\int_{\frac{j-1}{n}}^{\frac{j}{n}}\int_{\frac{i-1}{n}}^{\frac{i}{n}}
\left|f\left(i_n,j_n\right)-f\left(x,y\right) \right|\,dx\,dy\\
&\qquad \le \frac{2 \|f\|_{\infty} M2^{\alpha/2}}{n^\alpha},
\label{bothsq}
\end{align}
with the final inequality following from~\eqref{both} of the previous lemma. We note directly if $a<b$ then
\begin{align}
\frac{n^2}{h_{ab}^2} \sum_{(i,j)\in R_{ab}}
\int_{\frac{j-1}{n}}^{\frac{j}{n}}\int_{\frac{i-1}{n}}^{\frac{i}{n}}f^2\left(x,y\right)\,dx\,dy
&=\frac{n^2}{h_{ab}^2}
\iint_{\omega_{ab}}f^2\left(x,y\right)\,dx\,dy.
\label{offdsquare}
\end{align}
From the off-diagonal entries, for which $a<b$, we may conclude directly from~\eqref{bothsq} and~\eqref{offdsquare}
that
\begin{align*}
\left| \overline{ f^2}_{ab}-\widetilde{ f^2}_{ab}
\right|\le \frac{2 \|f\|_{\infty} M2^{\alpha/2}}{n^\alpha}.
\end{align*}
If on the other hand $a=b$ then
\begin{multline*}
\frac{n^2}{\tbinom{h_b}{2}} \sum_{(i,j)\in R_{bb}}
\int_{\frac{j-1}{n}}^{\frac{j}{n}}\int_{\frac{i-1}{n}}^{\frac{i}{n}}f^2\left(x,y\right)\,dx\,dy
\\ =\frac{n^2}{\tbinom{h_b}{2}}\sum_{j=(b-1)h+1}^{hb \I(b<k)+n \I(b=k)}\sum_{i=(b-1)h+1}^{j-1}
\int_{\frac{j-1}{n}}^{\frac{j}{n}} \int_{\frac{i-1}{n}}^{\frac{i}{n}}f^2\left(x,y\right)\,dx\,dy.
\end{multline*}
We shall start by simplifying this expression. We note that the latter becomes:
\begin{align}
\nonumber
\frac{n^2}{\tbinom{h_b}{2}}&\sum_{j=(b-1)h+1}^{hb \I(b<k)+n \I(b=k)}
\int_{\frac{j-1}{n}}^{\frac{j}{n}}\int_{\frac{(b-1)h}{n}}^{\frac{j-1}{n}}f^2\left(x,y\right)\,dx\,dy\\
&=\overline{ f^2}_{bb}+\frac{1}{(h_b-1)}\frac{2n^2}{h_b}\sum_{j=(b-1)h+1}^{hb \I(b<k)+n \I(b=k)}
\int_{\frac{j-1}{n}}^{\frac{j}{n}}\int_{\frac{j-1}{n}}^{y}
\left(\overline{ f^2}_{bb}-f^2(x,y)\right)\,dx\,dy.\nonumber
\end{align}
We may note directly that
\begin{equation*}
\frac{(b-1)h}{n}\le x<y\le \frac{b h\I(b<k)+n \I(b=k)}{n}=\frac{bh}{n} \I(b<k)+\I(b=k),
\end{equation*}
and so it follows that
\begin{align*}
\left|\overline{ f^2}_{bb}\right. &\left.-f^2(x,y)\right|\le \left|\frac{\int_{\frac{h(b-1)}{n}}^{\frac{bh}{n} \I(b<k)+\I(b=k) } \int_{\frac{(b-1)h}{n}}^{\frac{bh}{n} \I(b<k)+\I(b=k)}f^2(x',y')\,dx'\,dy'}{\left\{\frac{h_b}{n}\right\}^2}-f^2(x,y)\right|\\
&\le \frac{n^2}{h_b^2}\int_{\frac{(b-1)h}{n}}^{\frac{bh}{n} \I(b<k)+\I(b=k)} \int_{\frac{(b-1)h}{n}}^{\frac{bh}{n} \I(b<k)+\I(b=k) }\left|f^2(x',y')-f^2(x,y)\right|\,dx'\,dy'\\
&\le \frac{2 \|f\|_{\infty} n^2 }{h_b^2}\int_{\frac{(b-1)h}{n}}^{\frac{bh}{n} \I(b<k)+\I(b=k) } \int_{\frac{(b-1)h}{n}}^{\frac{bh}{n} \I(b<k)+\I(b=k) }\left|f(x',y')-f(x,y)\right|\,dx'\,dy'\\
&\le 2 \|f\|_{\infty}  M \left( \frac{\sqrt{2} h_b}{n}\right)^{\alpha}.
\end{align*}
Thus
\begin{align*}
&\left|\frac{n^2}{\tbinom{h_b}{2}}\sum_{j=(b-1)h+1}^{hb \I(b<k)+n \I(b=k)}
\int_{\frac{j-1}{n}}^{\frac{j}{n}}\int_{\frac{(b-1)h}{n}}^{\frac{j-1}{n}}f^2\left(x,y\right)\,dx\,dy
-\overline{ f^2}_{bb}\right|\\
&\qquad \le \frac{1}{(h_b-1)}\frac{2n^2}{h_b}\sum_{j=(b-1)h+1}^{hb \I(b<k)+n \I(b=k)}
\int_{\frac{j-1}{n}}^{\frac{j}{n}}\int_{\frac{j-1}{n}}^{y}
2 \|f\|_{\infty}  M \left( \frac{\sqrt{2} h_b}{n}\right)^{\alpha}\,dx\,dy.\\
&\qquad=\frac{1}{(h_b-1)}2 \|f\|_{\infty}  M \left( \frac{\sqrt{2} h_b}{n}\right)^{\alpha}.
\end{align*}
For the on-diagonal entries having $a=b$, it therefore follows that
\begin{align*}
\left| \overline{ f^2}_{bb}-\widetilde{ f^2}_{bb}
\right|\le \frac{2 \|f\|_{\infty} M2^{\alpha/2}}{n^\alpha}+\frac{1}{(h_b-1)}2 \|f\|_{\infty}  M \left( \frac{\sqrt{2} h_b}{n}\right)^{\alpha}.
\end{align*}
Note that $\frac{h_b^{\alpha}}{(h_b-1)}\le 2^{\alpha}$ if $h_b\ge 2$, and so the expression follows.
\end{proof}

\providecommand*\hyphen{-}

\end{document}